\PassOptionsToPackage{dvipsnames,svgnames,x11names}{xcolor}
\documentclass[12pt]{article}

\usepackage{amsmath,amssymb}
\usepackage{iftex}
\usepackage{lmodern}
\ifPDFTeX\else  
\fi
\IfFileExists{upquote.sty}{\usepackage{upquote}}{}
\IfFileExists{microtype.sty}{
  \usepackage[]{microtype}
  \UseMicrotypeSet[protrusion]{basicmath} 
}{}
\makeatletter
\@ifundefined{KOMAClassName}{
  \IfFileExists{parskip.sty}{%
    \usepackage{parskip}
  }{
    \setlength{\parindent}{0pt}
    \setlength{\parskip}{6pt plus 2pt minus 1pt}}
}{
  \KOMAoptions{parskip=half}}
\makeatother
\usepackage{xcolor}
\makeatletter
\ifx\paragraph\undefined\else
  \let\oldparagraph\paragraph
  \renewcommand{\paragraph}{
    \@ifstar
      \xxxParagraphStar
      \xxxParagraphNoStar
  }
  \newcommand{\xxxParagraphStar}[1]{\oldparagraph*{#1}\mbox{}}
  \newcommand{\xxxParagraphNoStar}[1]{\oldparagraph{#1}\mbox{}}
\fi
\ifx\subparagraph\undefined\else
  \let\oldsubparagraph\subparagraph
  \renewcommand{\subparagraph}{
    \@ifstar
      \xxxSubParagraphStar
      \xxxSubParagraphNoStar
  }
  \newcommand{\xxxSubParagraphStar}[1]{\oldsubparagraph*{#1}\mbox{}}
  \newcommand{\xxxSubParagraphNoStar}[1]{\oldsubparagraph{#1}\mbox{}}
\fi
\makeatother

\usepackage{longtable,booktabs,array}
\usepackage{calc} 
\usepackage{etoolbox}
\makeatletter
\patchcmd\longtable{\par}{\if@noskipsec\mbox{}\fi\par}{}{}
\makeatother
\IfFileExists{footnotehyper.sty}{\usepackage{footnotehyper}}{\usepackage{footnote}}
\makesavenoteenv{longtable}
\usepackage{graphicx}
\makeatletter
\def\maxwidth{\ifdim\Gin@nat@width>\linewidth\linewidth\else\Gin@nat@width\fi}
\def\maxheight{\ifdim\Gin@nat@height>\textheight\textheight\else\Gin@nat@height\fi}
\makeatother
\setkeys{Gin}{width=\maxwidth,height=\maxheight,keepaspectratio}
\makeatletter
\def\fps@figure{htbp}
\makeatother

\makeatletter
\@ifpackageloaded{caption}{}{\usepackage{caption}}
\AtBeginDocument{%
\ifdefined\contentsname
  \renewcommand*\contentsname{Table of contents}
\else
  \newcommand\contentsname{Table of contents}
\fi
\ifdefined\listfigurename
  \renewcommand*\listfigurename{List of Figures}
\else
  \newcommand\listfigurename{List of Figures}
\fi
\ifdefined\listtablename
  \renewcommand*\listtablename{List of Tables}
\else
  \newcommand\listtablename{List of Tables}
\fi
\ifdefined\figurename
  \renewcommand*\figurename{Figure}
\else
  \newcommand\figurename{Figure}
\fi
\ifdefined\tablename
  \renewcommand*\tablename{Table}
\else
  \newcommand\tablename{Table}
\fi
}
\@ifpackageloaded{float}{}{\usepackage{float}}
\floatstyle{ruled}
\@ifundefined{c@chapter}{\newfloat{codelisting}{h}{lop}}{\newfloat{codelisting}{h}{lop}[chapter]}
\floatname{codelisting}{Listing}

\makeatother
\makeatletter
\@ifpackageloaded{caption}{}{\usepackage{caption}}
\@ifpackageloaded{subcaption}{}{\usepackage{subcaption}}
\makeatother

\ifLuaTeX
  \usepackage{selnolig}  
\fi
\usepackage[]{natbib}
\usepackage{bookmark}

\IfFileExists{xurl.sty}{\usepackage{xurl}}{} 
\hypersetup{
  pdftitle={Title},
  pdfauthor={Author 1; Author 2},
  pdfkeywords={3 to 6 keywords, that do not appear in the title},
  colorlinks=true,
  linkcolor={blue},
  filecolor={Maroon},
  citecolor={Blue},
  urlcolor={Blue},
  pdfcreator={LaTeX via pandoc}}

\newcommand{\anon}{1}

\usepackage{graphicx} 
\usepackage{amsmath,amssymb,amsthm}
\usepackage{mathtools}
\usepackage{bm}
\usepackage{xspace}
\usepackage{xcolor}

\definecolor{vz}{rgb}{0.0, 0.0, 0.7}

\def\cW{\mathcal{W}}
\def\cA{\mathcal{A}}

\def\cG{\mathcal{G}}

\def\cQ{\mathcal{Q}}
\def\cH{\mathcal{H}}
\def\cM{\mathcal{M}}

\def\cF{\mathcal{F}}
\def\cY{\mathcal{Y}}

\def\cO{\mathcal{O}}
\def\cbO{\mathcal{\bar{O}}}

\def\ind{\perp\!\!\!\perp}
\def\given{\mathrel{|}}

\def\R{\mathbb{R}}

\def\bO{\bar{O}}
\def\bo{\bar{o}}

\def\bg{\mathbf{g}}

\def\eic{D}

\def\gbar{\bar{g}}

\def\Qinf{Q_{\infty}}
\def\bQ{\bar{Q}}
\def\bQinf{\bar{Q}_{\infty}}

\def\cbQ{\bar{\mathcal{Q}}}

\def\hpi{\hat{\pi}}
\def\hpibar{\overline{\hat{\pi}}}
\def\hpsi{\hat{\psi}}

\def\tsigma{\tilde{\sigma}}

\def\gbarinf{\bar{g}_{\infty}}

\def\pto{\overset{p}{\to}}
\def\dto{\overset{d}{\to}}

\def\gbarTMLE{\text{ADL-TMLE}\xspace}
\def\giTMLE{\text{AD-TMLE}\xspace}

\newcommand{\norm}[1]{\|#1\|}

\newcommand{\supcovering}[1]{N_{\infty}({#1},\bar{\mathcal{Q}})}
\newcommand{\supcoveringTwo}[2]{\mathcal{N}_{\infty}\left({#1},{#2}\right)}

\newcommand{\supintegralTwo}[2]{\mathcal{J}_{\infty}\left({#1},{#2}\right)}

\newcommand{\logit}{\mathop{\mathrm{logit}}}

\def\s{s}

\allowdisplaybreaks

\usepackage{xr}  
\usepackage[font=footnotesize]{caption}

\newtheorem{assumption}{Assumption}
\newtheorem{definition}{Definition}
\newtheorem{theorem}{Theorem}

\newtheoremstyle{myremarkstyle}
  {\topsep}   
  {\topsep}   
  {}          
  {}          
  {\bfseries} 
  {.}         
  { }         
  {}          

\theoremstyle{myremarkstyle}
\newtheorem{remark}{Remark}

\begin{document}

\def\spacingset#1{\renewcommand{\baselinestretch}%
{#1}\small\normalsize} \spacingset{1}


\if1\anon
{
  \title{\bf Efficient Statistical Estimation for Sequential Adaptive Experiments with Implications for Adaptive Designs}
  \author{Wenxin Zhang, Mark van der Laan\\
    Division of Biostatistics, University of California, Berkeley}
  \maketitle
} \fi

\if0\anon
{
  \bigskip
  \bigskip
  \bigskip
  \begin{center}
    {\LARGE\bf Efficient Statistical Estimation and Optimal Design for Non-i.i.d. and Possibly Adaptive Experiments}
\end{center}
  \medskip
} \fi

\bigskip
\begin{abstract}
Adaptive experimental designs have gained popularity in clinical trials and online experiments. Unlike traditional, fixed experimental designs, adaptive designs can dynamically adjust treatment randomization probabilities and other design features in response to data accumulated sequentially during the experiment. These adaptations are useful to achieve diverse objectives, including reducing uncertainty in the estimation of causal estimands or increasing participants' chances of receiving better treatments during the experiment.
At the end of the experiment, it is often desirable to answer causal questions from the observed data.
However, the adaptive nature of such experiments and the resulting dependence among observations pose significant challenges to providing valid statistical inference and efficient estimation of causal estimands. Building upon the Targeted Maximum Likelihood Estimation (TMLE) framework tailored for adaptive designs \citep{van2008construction}, we introduce a new Adaptive-Design-Likelihood-based TMLE (ADL-TMLE) to estimate a wide class of causal estimands from adaptive experiment data, including the average treatment effect as our primary example. 
We establish asymptotic normality and semiparametric efficiency of \gbarTMLE under relaxed positivity and design stabilization assumptions for adaptive experiments. Motivated by these results, we further propose a novel adaptive design aimed at minimizing the variance of the estimator based on data generated under that design.
Simulations show that \gbarTMLE demonstrates superior variance-reduction performance across different types of adaptive experiments, and that the proposed adaptive design attains lower variance than the standard efficiency-oriented adaptive design. Finally, we generalize our framework to broader settings, including those with longitudinal structures.
\end{abstract}

\noindent%
\vfill

\newpage
\spacingset{1.6} 


\section{Introduction}
Adaptive experimental designs are useful for clinical trials due to their flexibility and statistical efficiency \citep{chow2008adaptive}. Unlike traditional, fixed trial designs, adaptive designs allow investigators to sequentially modify treatment randomization probabilities and/or other key design features based on accumulated data during the experiment. These adaptations can be targeted for a range of objectives: such as improving estimation efficiency for causal estimands and increasing participants’ chances of receiving superior treatments.

Response-Adaptive Randomization (RAR), originating from the foundational work of \cite{thompson1933likelihood}, is an important adaptive design procedure in which treatment assignment probabilities are updated based on prior participants' treatments and outcomes \citep{hu2006theory,robertson2023response}. Covariate-Adjusted Response-Adaptive (CARA) designs extend RAR designs by incorporating baseline covariates into treatment assignment probabilities, enabling personalized treatment allocation for improved participant welfare or statistical efficiency \citep{rosenberger2001covariate,hu2006theory,zhang2007asymptotic,van2008construction,rosenberger2008handling}. The adaptive design approaches are also closely related to the literature on multi-armed bandits and contextual bandits, with broad applications in real-life decision making \citep{bubeck2012regret,lattimore2020bandit,bouneffouf2020survey}.

While adaptive experiments have advantage of improving estimation efficiency and/or bringing benefits to participants compared with those with fixed, non-adaptive designs, they also present challenges for statistical inference. At the end of an adaptive experiment, investigators often seek to estimate causal estimands based on the observed data, such as the average treatment effect or mean outcome under personalized treatment rules. However, the adaptivity of such experiments induces dependence among observations, violating the i.i.d. assumptions that standard statistical inference methods rely on. This has motivated substantial methodological work to provide statistical inference for data collected under adaptive designs. 
For example, \citet{zhang2007asymptotic} and \citet{zhu2015covariate} derived asymptotic results for inference under CARA designs within parametric modeling frameworks. \citet{van2008construction} introduced Targeted Maximum Likelihood Estimation (TMLE) framework which is tailored for adaptive design data and provides non-parametric statistical inference method without relying on parametric assumptions. Subsequent developments include the development of statistical inference for causal estimands such as the risk difference and log-relative risk \citep{chambaz2014inference} as well as for estimating the mean outcome under optimal treatment rules \citep{murphy2003optimal,luedtke2016super} in CARA experiments \citep{chambaz2017targeted}. 
Another line of work employs variants of stabilized reweighting techniques \citep{luedtke2016statistical} and develops asymptotically normal estimators for causal estimands in non‑contextual \citep{hadad2021confidence} and contextual bandit settings \citep{bibaut2021post,zhan2021off}. 

In contrast to these methods, we develop a unified framework that enables semiparametrically efficient estimation and inference for adaptive experiments and provides a new adaptive design that targets estimation efficiency.

First, we introduce a novel Adaptive-Design-Likelihood-based TMLE (\gbarTMLE) framework that enables statistical estimation and inference for a wide class of causal estimands from adaptive experiments, including the average treatment effect as our primary example.
This estimator is driven by the likelihood structure of adaptive experiment and solves the score equation for components that exclusively define the causal estimand of interest.
We show that \gbarTMLE is doubly robust, asymptotically normal and semiparametrically efficient.

A key insight from these results is that the minimal asymptotic variance of regular estimators under the adaptive experiment is inherently linked to that under a reference i.i.d. experiment in which treatments are assigned by the average of treatment randomization functions across the adaptive experiment (which we refer to as the ``average design'').
This highlights the advantage of \gbarTMLE in enabling statistical inference under relaxed assumptions for adaptive experiments.
Specifically, it suffices to require positivity only for the average design, rather than for every individual treatment assignment probability.
This enables greater flexibility in adaptive experiments, where local positivity violations are allowed during some segments of the experiment.
The design stabilization conditions for establishing asymptotic normality and efficiency are also relaxed, as they are required only for the average design, which also leads to improved finite-sample performance of \gbarTMLE.
Moreover, in scenarios where individual-level treatment randomization mechanisms are inaccessible due to operational or privacy constraints, our estimation framework remains feasible as long as the average design is accessible.
These theoretical and practical advantages are further supported by simulation studies, which demonstrate the superior variance-reduction performance of ADL-TMLE across different adaptive experiments with varied objectives.

Second, while our framework provides a semiparametrically efficient estimator for a given adaptive experiment, even greater efficiency can be attained by guiding the real-time adaptive design toward an oracle design that is itself optimized for estimation efficiency.
Motivated by \gbarTMLE framework, we propose a novel adaptive design framework that dynamically steers the average design toward the evolving estimate of the oracle design as data accumulate throughout the experiment.
This adaptive procedure leads to a lower variance than the standard efficiency-oriented adaptive design as shown by simulations.

Finally, our framework is directly applicable to settings with independent data structures, such as cases where certain factors induce heterogeneity in the treatment assignment mechanism but do not affect the primary outcome conditional on treatment and other baseline covariates. Furthermore, we generalize our framework to broader settings, including those with longitudinal data structures.

The remainder of the paper is organized as follows. Section \ref{section:setup} introduces the statistical setup. Section \ref{section:estimation} provides the estimation procedure of \gbarTMLE. Sections \ref{section:asy_normality} and \ref{section:efficiency} establish asymptotic normality and semiparametric efficiency. Section \ref{section:adaptive_design} describes the proposed adaptive design approach. Section \ref{section:simulation} reports simulation studies evaluating the performance of the proposed estimation and design methods. Section \ref{section:generalization} discusses applications and generalizations of our framework. Finally, Section \ref{section:conclusion} concludes.

\section{Data, Model and Statistical Setup}
\label{section:setup}

We define the full data as $X = (W,Y(1),Y(0)) \sim P^F \in \cM^F$, where $W \in \cW$ are baseline covariates, $Y(a) \in \cY$ is the counterfactual outcome under treatment level $a \in \cA$. Without loss of generality, we consider binary treatment space $\cA = \{0,1\}$ and set $\cY = [0,1]$.
Here $P^F$ is the true distribution that generates full data $X$, living in a general model $\cM^F$ without model prespecification. We use $E_{P^F}$ to denote the expectation under $P^F$.
While our framework is generally applicable for estimating a wide class of causal estimands (see Section \ref{section:generalization}), here we use ATE $\Psi(P^F) := E_{P^F}[Y(1)-Y(0)]$ as the main example for illustration.

Consider an adaptive experiment that enrolls $n$ units sequentially, where $i$ is an index of each unit reflecting their enrollment order and $O_i := (W_i,A_i,Y_i)$ are the data to be observed for this unit. 
Upon the enrollment, the experimenter first observes baseline covariates $W_i$, then assigns treatment $A_i$ following a treatment randomization function $g_i:\cW \times \cA \to [0,1]$ within a function class $\cG$. Then, the unit's outcome $Y_i$ is observed. 
The adaptivity of the experiment is reflected in that treatment $A_i$ is determined by a function $g_i$ that is itself generated based on the observed data of previously enrolled units before $i$, which we denote by $\bO(i-1) := (O_1,\cdots,O_{i-1})$.
We use $\bg_{1:n} := (g_1,\cdots, g_i, \cdots, g_n)$ to refer to an adaptive design that includes a collection of treatment randomization functions $g_i \in \cG$ from $i = 1,\cdots,n$, which are known and controlled by the experimenter.

Define $\cQ := \cQ_W \otimes \cQ_Y$, where $\cQ_W$ is the set of possible distributions of $W$ and $\cQ_Y$ is the set of possible conditional distributions of $Y$ given $A$ and $W$. 
For any $Q = (Q_W, Q_Y) \in \cQ$, we denote $\bQ(A,W)$ as a shorthand of $E[Y|A,W]$, the conditional mean of $Y$ given $A$ and $W$.
For every $i \in [n]$, we assume that $W_i \sim Q_W$ and $Y_i \sim Q_Y(\cdot |A_i,W_i )$. 
We use $P_Q^n \in \cM^n$ to denote the distribution that generates the adaptive experiment data $\bO(n)$ in the adaptive design model $\cM^n$. 
The model specifies the likelihood of $\bo(n)=(o_1,\cdots,o_i,\cdots,o_n)$ as 
\[
p^n_{q,g_1,\cdots,g_n}\left(\bo(n)\right)
=  \prod_{i=1}^{n}q_{w}(w_i) 
\times
\prod_{i=1}^{n} g_{i}(a_i|w_i)
\times \prod_{i=1}^{n}q_{y}(y_i|a_i,w_i),
\]
where we denote $q_{w}$ and $q_{y}$ as the probability density (or mass) functions corresponding to the distributions of $Q_W$ and $Q_Y$, and use $g_i(a|w)$ is a shorthand of $g_i(A=a|W=w)$.

One can identify the causal estimand $\Psi(P^F)$ with the following assumptions:
\begin{assumption}[Sequential Exchangeability]\label{assumption:sequential_randomization} 
$\forall i \in [n]$, $Y_{i}(a) \ind A_i \given 
 W_i, \bO(i-1)$.
\end{assumption}

\begin{assumption}[Theoretical Positivity]\label{assumption:identification_positivity} 
$
\forall w \in \mathcal{W}, \, a \in \mathcal{A}, \, \exists i \in [n] \text{ such that } g_{i}(a|w) > 0.$
\end{assumption}

\begin{theorem}[Identification of ATE]
With Assumptions \ref{assumption:sequential_randomization} and \ref{assumption:identification_positivity}, $\Psi(P^F)$ is identified by
\begin{eqnarray*}
    \Psi(P^F) 
    &=& E[\bQ(1,W)-\bQ(0,W)].
\end{eqnarray*}
\end{theorem}
We denote this target estimand by $\psi(Q):= E[\bQ(1,W)-\bQ(0,W)]$.

\section{Targeted Maximum Likelihood Estimation}
\label{section:estimation}
At the end of an adaptive experiment, one observes a single series of dependent data $\bO(n) = (O_1,\cdots,O_n)$. 
Let $\bQ_n$ denote an initial estimate of the true conditional mean function $\bQ_0$, estimated by a user-specified model or an ensemble learner such as Super Learner \citep{van2007super}.
For any adaptive design $\bg_{1:n}$, we define its average design $\gbar_n := \frac{1}{n}\sum_{i=1}^n g_i$.

Our proposed \gbarTMLE is
\[
\psi(Q_n^*) := \frac{1}{n}\sum_{i=1}^{n} \left[\bQ_n^*(1,W_i) - \bQ_n^*(0,W_i)\right],
\]
where $Q_n^*=(Q_{n,W}, Q_{n,Y}^*)$ with $Q_{n,W}$ denotes the empirical distribution of baseline covariates, and $Q_{n,Y}^*$ an updated outcome distribution associated with a TMLE update $\bQ_n^*$ of the initial estimate $\bQ_n$. Specifically, $\bQ_n^*$ is obtained by solving
\[
\frac{1}{n} \sum_{i=1}^{n} \frac{2A_i-1}{\gbar_n(A_i|W_i)}\left[Y_i-\bQ_n^*(A_i,W_i)\right] = 0,
\]
via the following regression:
\[
\logit \bQ_{n,\epsilon} \left(A, W \right) = \logit \bQ_{n} \left(A,W \right) + \epsilon H_n(A,W),
\]
where $H_n$ is defined by $H_n(A,W):=\frac{2A-1}{\gbar_n(A|W)}$ based on the average design $\gbar_n$. The TMLE update is $\bQ_n^*=\bQ_{n,\epsilon_n}$, where $\epsilon_n$ is the solution of that regression.

Define
$
\sigma_n^2 := \frac{1}{n}\sum_{i=1}^n\left\{ \frac{2A_i-1}{\gbar_n(A_i|W_i)}\left[Y_i-\bQ_n^*(A_i,W_i)\right] + \bQ_n^*(1,W_i) - \bQ_n^*(0,W_i) - \psi(Q_n^*)\right\}^2$.
The variance estimator of $\psi(Q_n^*)$ is given by $\sigma_n^2/n$, 
The corresponding $100(1 - \alpha)\%$ confidence interval is $[\psi(Q_n^*)-z_{1-\alpha/2}\sigma_n/\sqrt{n}, \psi(Q_n^*)+z_{1-\alpha/2}\sigma_n/\sqrt{n}]$,
where $z_{1-\alpha/2}$ is the $1-\alpha/2$ quantile of standard normal distribution (e.g. $\alpha = 0.05$).

We note that our proposed \gbarTMLE is driven by the likelihood structure of adaptive experiment and solves the score equation for the $Q$-components that exclusively define the target parameter $\psi(Q)$.
For any $Q = (Q_W, Q_Y) \in \cQ$ and treatment randomization function $g \in \cG$, we denote $D(Q,g)$ as the efficient influence curve (EIC) of the target estimand in an i.i.d. experiment where $W \sim Q_W, A \sim g(\cdot | W)$, and $Y\sim Q_Y(\cdot| A,W )$. For ATE, the EIC is 
$
\eic(Q,g)(O) = \eic_1(Q,g)(O)+\eic_2(Q)(O),
$ where $\eic_1(Q,g)(O):=\frac{2A-1}{g(A|W)}\left(Y-\bQ(A,W)\right)$ and $\eic_2(Q)(O):= \bQ(1,W) - \bQ(0,W) - \psi(Q)$.
The core procedure of \gbarTMLE is to solve $\frac{1}{n}\sum_{i=1}^n D(Q,\gbar_n)(O_i)=0$, which is equivalent to solving the score equation as if in the i.i.d. experiment where the treatment distribution follows $A \sim \gbar_n(\cdot | W)$, even though the data is collected from the adaptive experiment. 

For comparison, we also introduce the original TMLE tailored for adaptive design data proposed by \cite{van2008construction}, which we denote by $\psi(Q_n^{\dagger})$ and refer to as ``\giTMLE''  (adaptive-design TMLE). In \giTMLE, $Q_n^{\dagger}=(Q_{n,W},Q_{n,Y}^{\dagger})$, and $\bQ_n^{\dagger}$ is the TMLE update of the estimated conditional mean function $\bQ_n$, which solves 
$
\frac{1}{n} \sum_{i=1}^{n} \frac{2A_i - 1}{g_i(A_i|W_i)} [Y_i - \bQ_n^{\dagger}(A_i,W_i)] = 0$ using every $i$-specific treatment randomization function $g_i$ from the adaptive design $\bg_{1:n}$. 
Its variance estimator is given by $\tsigma_n^2/n$, where
$
\tsigma_n^2 := \frac{1}{n}\sum_{i=1}^n\left\{ \frac{2A_i-1}{g_i(A_i|W_i)}\left[Y_i-\bQ_n^{\dagger}(A_i,W_i)\right] + \bQ_n^{\dagger}(1,W_i) - \bQ_n^{\dagger}(0,W_i) - \psi(Q_n^{\dagger})\right\}^2$.
While next sections focus on properties of TMLEs, they also apply to the corresponding One-Step/AIPW estimators
$
\hpsi_n(\bQ_n) := \frac{1}{n} \sum_{i=1}^{n} \frac{2A_i-1}{\gbar_n(A_i|W_i)}\left[Y_i-\bQ_n(A_i,W_i)\right] + \frac{1}{n}\sum_{i=1}^{n} \left[\bQ_n(1,W_i) - \bQ_n(0,W_i)\right],
$
and 
$
\hpsi_n^{\dagger}(\bQ_n):=\frac{1}{n} \sum_{i=1}^{n} \frac{2A_i-1}{g_i(A_i|W_i)}\left[Y_i-\bQ_n(A_i,W_i)\right] + \frac{1}{n}\sum_{i=1}^{n} \left[\bQ_n(1,W_i) - \bQ_n(0,W_i)\right]$.

\section{Asymptotic Normality}
\label{section:asy_normality}

In this section we analyze the asymptotic normality of \gbarTMLE.
We define $P_{Q,g}f:=\int f(o) q(o)g(o)d\mu(o)$ for any $Q$ and $g$, and $P_{Q,g_i}f:=\int f(o) q(o)g_i(o)d\mu(o)$ for every $g_i$ in adaptive design, where $\mu$ is a measure on $\cO$ equipped with the product of the $\sigma$-fields of $\cW$, $\cA$ and $\cY$. Note that $P_{Q,\gbar_n}f=\frac{1}{n}\sum_{i=1}^nP_{Q,g_i}f$.

To establish the asymptotic normality of $\psi(Q_n^*)$, we first decompose the difference between $\psi(Q_n^*)$ and $\psi(Q_0)$ as follows.

\begin{theorem}[Difference between $\psi(Q_n^*)$ and $\psi(Q_0)$]
\label{theorem:tmle_decomposition}
The difference of $\psi(Q_n^*)$ and the true parameter $\psi(Q_0)$ can be decomposed as
\begin{eqnarray*}
&& \psi(Q_n^*) - \psi(Q_0)
= M_{1,n}(\Qinf,\gbarinf) + M_{2,n}(Q_n^*,\Qinf,\gbar_n) + M_{3,n} (\Qinf,\gbar_n,\gbarinf), \label{eq:TMLE_decomposition} 
\end{eqnarray*}
where
\begin{eqnarray*}
M_{1,n}(\Qinf,\gbarinf) 
&:=& \frac{1}{n} \sum_{i=1}^n \left[\eic(\Qinf,\gbarinf)(O_i)
- P_{Q_0,g_i} \eic(\Qinf,\gbarinf)\right], \\
M_{2,n}(Q_n^*,\Qinf,\gbar_n) &:=&\frac{1}{n} \sum_{i=1}^n  \Big\{ \left[\eic(Q_n^*,\gbar_n)(O_i)
- P_{Q_0,g_i} \eic(Q_n^*,\gbar_n)\right]\\ 
&& - \left[\eic(\Qinf,\gbar_n)(O_i)
- P_{Q_0,g_i} \eic(\Qinf,\gbar_n) \right] \Big\}, \\
M_{3,n}(\Qinf,\gbar_n,\gbarinf) &:=& \frac{1}{n} \sum_{i=1}^n \bigg\{\left[\eic(\Qinf,\gbar_n)(O_i)-P_{Q_0,g_i} \eic(\Qinf,\gbar_n)\right] \\
&& - \left[\eic(\Qinf,\gbarinf)(O_i) -P_{Q_0,g_i} \eic(\Qinf,\gbarinf)\right] \bigg\},
\end{eqnarray*}
with any fixed $\gbarinf \in \cG$ and $\Qinf:=(Q_{\infty,W}, Q_{\infty,Y}) \in \cQ$, where $Q_{\infty,W}=Q_{0,W} \in \cQ_W$ and  $Q_{\infty,Y} \in \cQ_Y$ with its conditional mean function of the outcome $\bQinf \in \cbQ$.
\end{theorem}
The proof of this decomposition is given in Appendix \ref{app:proofs}. 

Suppose there exist fixed $\Qinf$ and $\gbarinf$ such that $\bQ_n^*$ and $\gbar_n$ converge to them, respectively, under the assumptions stated below.
We then establish the asymptotic normality of $M_{1,n}$, which is the average of a martingale difference sequence indexed by $\Qinf$ and $\gbarinf$.

\begin{assumption}[Boundedness of the Average of Adaptive Design]
\label{assumption:strong_positivity_of_average_design}
    There exists $\zeta > 0$ such that $\inf_{a \in \cA, \, w \in \cW}\gbar_n(a|w) \ge \zeta$. 
\end{assumption}

\begin{assumption}[Stabilized Variance under the Average of Adaptive Design] 
\label{assumption:stabilized_variance}
    There exists $\sigma^2_{0} \in \R^+$ such that 
    $
    P_{Q_0, \gbar_n} \eic(\Qinf,\gbarinf)^2 \pto \sigma^2_{0}.$
\end{assumption}

We also provide an alternative to Assumption \ref{assumption:stabilized_variance} that characterizes how the average design $\gbar_n$ converges to the fixed design $\gbarinf$.

\begin{assumption}
\label{assumption:stabilized_variance_gbarinf}
    $P_{Q_0, \gbar_n} \eic(\Qinf,\gbarinf)^2 \pto P_{Q_0, \gbarinf} \eic(\Qinf,\gbarinf)^2$.
\end{assumption}

\begin{theorem}
[Asymptotic Normality of $M_{1,n}$] 
\label{theorem:asynormal_first}
    Suppose assumptions \ref{assumption:strong_positivity_of_average_design} and \ref{assumption:stabilized_variance} (or \ref{assumption:stabilized_variance_gbarinf}) hold, then
    \[
    \sqrt{n}M_{1,n}(\Qinf,\gbarinf) \dto N(0,\sigma_0^2).
    \]
\end{theorem}

This asymptotic normality directly follows the martingale central limit theorems (e.g. Theorem 2 in \cite{brown1971martingale}). 

For the remaining martingale process terms $M_{2,n}$ and $M_{3,n}$, we establish their equicontinuity conditions following the equicontinuity results of the general setup in \cite{zhang2024evaluating}.
\begin{assumption}[Covering Number of $\cbQ$]
\label{assumption:reasonable_covering_integral}
$\cbQ$ is a class of functions with finite sup-norm covering number $\supcovering{\epsilon}$, with $\supintegralTwo{\epsilon}{\cbQ}:= \int_0^\epsilon \sqrt{\log \left(1+\supcoveringTwo{u}{\cbQ} \right)} d u$ the $\epsilon$-entropy integral in supremum norm of $\cbQ$.
\end{assumption}

\begin{assumption}[Convergence of $\bQ_n^*$] 
\label{assumption:sigma_N_convergence}
Define
\[
\s_{n}\left(\bQ, \bQinf, \gbar_n \right):=\sqrt{\frac{1}{n} \sum_{i=1}^n E \left[\left. \rho\left(\frac{| \eic_1\left(\bQ, \gbar_n\right)(O_i)-\eic_1\left(\bQinf,\gbar_n\right)(O_i)|}{c_1}\right) \right| \bO(i-1) \right]},
\]
where $\rho(x)=e^x-x-1$ and $c_1 = 8/\zeta$ with $\zeta$ as assumed in Assumption \ref{assumption:strong_positivity_of_average_design}.
Assume that $\s_n(\bQ_{n}^*,\bQinf,\gbar_n) = o_P(1)$.
\end{assumption}

\begin{theorem}[Equicontinuity of $M_{2,n}$]
\label{theorem:equicontinuity}
Suppose Assumptions 
\ref{assumption:strong_positivity_of_average_design},
\ref{assumption:reasonable_covering_integral},
\ref{assumption:sigma_N_convergence}
hold.
Then, $M_{2,n}(Q_n^*,\Qinf,\gbar_n) = o_P(n^{-\frac{1}{2}})$.
\end{theorem} 

\begin{assumption}[Covering Number of $\cG$]
\label{assumption:reasonable_covering_integral_gbar}
$\cG$ is a class of functions with finite sup-norm covering number $\supcoveringTwo{\epsilon}{\cG}$, with $\supintegralTwo{\epsilon}{\cG}:= \int_0^\epsilon \sqrt{\log \left(1+\supcoveringTwo{u}{\cG} \right)} d u$ the $\epsilon$-entropy integral in supremum norm of $\cG$.
\end{assumption}

\begin{assumption}[Convergence of Average Design]
\label{assumption:sigma_N_convergence_gbar}
Suppose $\inf_{a \in \cA,w \in \cW}\gbarinf(a|w) \ge \eta$ for some $\eta >0$.
Define
\[
v_n\left(\bQinf, \gbar_n, \gbarinf \right):=\sqrt{\frac{1}{n} \sum_{i=1}^n E \left[\left. \rho\left(\frac{| \eic_1\left(\bQinf, \gbar_n\right)(O_i)-\eic_1\left(\bQinf,\gbarinf\right)(O_i)|}{c_2}\right) \right| \bO(i-1) \right]},
\]
where $\rho(x)=e^x-x-1$ and $c_2 = 8/{\eta}^2$.
Assume that $v_n(\bQinf, \gbar_n, \gbarinf) = o_P(1)$.
\end{assumption}

\begin{theorem}[Equicontinuity of $M_{3,n}$]
\label{theorem:equicontinuity_gbar}
Suppose Assumptions
\ref{assumption:strong_positivity_of_average_design},
\ref{assumption:reasonable_covering_integral_gbar},
\ref{assumption:sigma_N_convergence_gbar}
hold.
Then, $M_{3,n}(\Qinf,\gbar_n, \gbarinf) = o_P(n^{-\frac{1}{2}})$.
\end{theorem} 

Combining the asymptotic normality and equicontinuity conditions,
we establish the asymptotic normality of the proposed \gbarTMLE $\psi(Q_n^*)$.

\begin{theorem}[Asymptotic Normality of \gbarTMLE]
\label{theorem:asympotic_normality_proposed_tmle}
    Suppose Assumptions \ref{assumption:strong_positivity_of_average_design},\ref{assumption:stabilized_variance} (or \ref{assumption:stabilized_variance_gbarinf}),  \ref{assumption:reasonable_covering_integral}, \ref{assumption:sigma_N_convergence},
    \ref{assumption:reasonable_covering_integral_gbar}, and \ref{assumption:sigma_N_convergence_gbar} hold, 
    then $\psi(Q_n^*)$ is asymptotically normal:
    \[
    \sqrt{n} \left[\psi(Q_n^*) - \psi(Q_0) \right] \dto N(0,\sigma_0^2).
    \]

\end{theorem}
Note that our assumptions do not require that $\bQ_n^*$ converges to a $\bQinf$ that equals the true conditional mean function $\bQ_0$. Thus  \gbarTMLE has the double robustness property despite using a misspecified estimate of $\bQ_0$, since the average design $\gbar_n$ is known.

\section{Semiparametric Efficiency}
\label{section:efficiency}
In this section, we provide semiparametric efficiency results of \gbarTMLE for adaptive experiments.
We start by defining a Hilbert space: 
$
\cH(Q,g) := \{h \in L_2^0(P_{Q,g}): h(o)=h_w(w)+h_{y}(w,a,y),\int h_w(w)dQ_w(w)  = 0, \int h_y(w,a,y)dQ_y(y|a,w)  = 0\}
$ for any $Q = (Q_W,Q_Y)\in \cQ$, with a reference treatment randomization distribution $g \in \cG$ and the corresponding density $p_{Q,g}(o)=q_w(w) g(a|w) q_y(y|a,w)$ under $P_{Q,g}$. For any $h_1,h_2 \in \cH(Q,g)$, define their inner product as $<h_1,h_2>_{\cH(Q,g)} := \int h_1(o)h_2(o)dP_{Q,g}(o)$.
Define a class of paths: 
\[\{Q^{h}_{\epsilon}=(Q_{w,\epsilon}^h,Q_{y,\epsilon}^h):dQ_{w,\epsilon}^h=(1+\epsilon h_w) dQ_w, dQ_{y,\epsilon}^h=(1+\epsilon h_y) dQ_y, h \in \cH(Q, \gbarinf)\}
\]
for any $Q \in \cQ$. 
We denote by $P^{n}_{Q_{\epsilon_n}^h}$ the fluctuated distribution of $P^{n}_{Q}$ with density $p^{n}_{Q_{\epsilon_n}^h,g_1,\cdots,g_n}$, where $\epsilon_n = 1/\sqrt{n}$. 

In the following, we establish the semiparametric efficiency of \gbarTMLE in adaptive design settings by showing that it achieves the minimal asymptotic variance among regular estimators in the associated asymptotically normal experiment, as implied by the convolution theorem (Theorem 3.11.2 in \cite{van1996weak}). The proofs are in Appendix \ref{app:proofs}.

\subsection{Asymptotically Normal Experiment}
We first analyze the sequence of statistical experiments $(\cbO(n),\cF_n, P^{n}_{Q_{\epsilon_n}^h}:h \in \cH(Q,\gbarinf))$ and show that it is asymptotically normal. 
    Let $\Delta_{n,h} = \frac{1}{\sqrt{n}}\sum_{i=1}^n h(O_i) + o_{P_Q^n}(1)$, and $\norm{h}_{\cH(Q,\gbarinf)}^2 = \int h^2(o) dP_{Q,\gbarinf}(o)$.
    Let $\{\Delta_h:h \in \cH(Q,\gbarinf)\}$ be an iso-Gaussian process with zero mean and covariance $P_{Q,\gbarinf} \Delta_{h_1}\Delta_{h_2}=<h_1, h_2>_{\cH(Q,\gbarinf)}$.
\begin{assumption}[Hilbert Space Norm Stability under Adaptive Design]
\label{assumption:Hilbert_space_approximation}
    For any $h \in \cH(Q,\gbarinf)$, $\norm{h}^2_{\cH(Q,\gbar_n)} = \norm{h}^2_{\cH(Q,\gbarinf)} + o_{P^n_{Q}}(1)$.
\end{assumption}
\begin{theorem}[Asymptotically Normal Experiment]
\label{theorem:our_LAN}
Suppose Assumption \ref{assumption:Hilbert_space_approximation} holds. Then, the experiment is asymptotically (shift) normal, i.e., for each $h \in \cH(Q_0,\gbarinf)$, 
    \[\log\frac{dP^{n}_{Q_{\epsilon_n}^h}(O_1,\cdots,O_n)}{dP^n_{Q}(O_1,\cdots,O_n)} = \Delta_{n,h}-\frac{1}{2}\norm{h}_{\cH(Q,\gbarinf)}^2 
    \]
    such that  $\Delta_{n,h} \dto \Delta_{h}$ under $P_Q^n$, and $\Delta_{h} \sim N(0, \norm{h}_{\cH(Q,\gbarinf)}^2)$.
\end{theorem}

\subsection{Regularity and Efficiency}
Now we introduce the definition of regular estimator and present the limit distribution for regular estimators of regular parameters based on the convolution theorem.
Define a sequence of parameters $\psi(Q_{\epsilon_n}^{h}):=E_{Q_{w,\epsilon_n}^h}[\bQ_{y,\epsilon_n}^h(1,W)-\bQ_{y,\epsilon_n}^h(0,W)]$ for every $h \in \cH(Q,\gbarinf)$, where $Q_{w,\epsilon_n}^h$ denotes the perturbed distribution of $Q_W$ and $\bQ_{y,\epsilon_n}^h$ is the conditional mean function of the outcome under the perturbed distribution of $Q_Y$. 

For every $h \in \cH(Q,\gbarinf)$, the sequence of parameters $\psi(Q_{\epsilon_n}^{h})$ is regular in that
\[
\sqrt{n} \left(\psi(Q_{\epsilon_n}^{h})-\psi(Q)\right) \to \dot{\psi}(Q)(h),
\]
for a continuous, linear map $\dot{\psi}: \cH(Q,\gbarinf) \to \R$ and a certain norming operator $\sqrt{n}$.
Here, $\dot{\psi}(Q)(h)$ is the pathwise derivative of $\psi(Q_{\epsilon}^{h})$ at $\epsilon = 0$ and equals $<D(Q,\gbarinf), h>_{\cH(Q,\gbarinf)}$ by Riesz representation.

\begin{definition}[Regular Estimator]
    If a sequence of estimators $T_n(O_1,\cdots,O_n)$ satisfies
    \[
    \sqrt{n}\left(T_n - \psi(Q_{\epsilon_n}^{h})\right) \dto L
    \]
    under $P^n_{Q_{\epsilon_n}^h}$ for every $h \in \cH(Q,\gbarinf)$ for a fixed, tight Borel probability measure $L$ on $\R$, then $T_n$ is regular with respect to the norming operators $\sqrt{n}$.
\end{definition}

\begin{theorem}[Limit Distribution of Regular Estimators via the Convolution Theorem]
\label{theorem:LAN_conv}
    Suppose Assumption \ref{assumption:Hilbert_space_approximation} holds such that the sequence of statistical experiments $(\cbO(n),\cF_n, P^{n}_{Q_{\epsilon_n}^h}:h \in \cH(Q,\gbarinf))$ is asymptotically normal. Then, for regular estimators $T_n$ of $\psi(Q)$, the limit distribution of  $\sqrt{n}(T_n - \psi(Q))$ equals the distribution of a sum $Z+B$ of independent and tight random variables such that 
    \[
    Z \sim N(0, \norm{\dot{\psi}^*(Q)}_{\cH(Q,\gbarinf)}^2),
    \]
    where $\dot{\psi}^*(Q)$ is the adjoint of $\dot{\psi}(Q)$ such that $<\dot{\psi}^*(Q),h>_{\cH(Q,\gbarinf)}=\dot{\psi}(Q)(h)$.
\end{theorem}

We establish the regularity and efficiency of \gbarTMLE as follows. 
\begin{theorem}[Regularity and Efficiency]\label{theorem:efficiency}
Suppose Assumptions \ref{assumption:strong_positivity_of_average_design},
\ref{assumption:stabilized_variance_gbarinf},  \ref{assumption:reasonable_covering_integral}, \ref{assumption:sigma_N_convergence},
    \ref{assumption:reasonable_covering_integral_gbar}, \ref{assumption:sigma_N_convergence_gbar} and
    \ref{assumption:Hilbert_space_approximation} hold with $\Qinf=Q_0$, then
    $\psi(\bQ_n^*)$ is regular and efficient, i.e., the asymptotic variance of $\sqrt{n} (\psi(Q_n^*)-\psi(Q_0))$ achieves the minimal asymptotic variance $\norm{\dot{\psi}^*(Q_0)}_{\cH(Q_0,\gbarinf)}^2$ among all regular estimators.
\end{theorem}

\begin{remark}[Comparison of two TMLEs]
We compare \gbarTMLE $\psi(Q_n^*)$ with \giTMLE $\psi(Q_n^{\dagger})$ based on the results of Sections \ref{section:asy_normality} and \ref{section:efficiency}.
First, to construct \giTMLE $\psi(Q_n^{\dagger})$, the positivity assumption is: $\inf_{a \in \cA, w \in \cW}g_i(a|w) \geq \zeta$ for every single $g_i$ ($i = 1, \cdots, n$).
In comparison, the positivity assumption for the proposed \gbarTMLE $\psi(Q_n^*)$ is less restrictive, as it is only required on the overall average design $\gbar_n$ instead of every $g_i$.
This enables a more flexible adaptive design for conducting statistical inference. For example, one can consider a design which allows more or less exploration and exploitation of treatment arms across different time points during the experiment and even permits cases where positivity is locally violated for some units, as long as the overall average design does not break the positivity assumption. Moreover, the establishment of \gbarTMLE only requires accessibility of the average design $\gbar_n$, making the estimation procedure feasible when reporting every single $g_i$ is impractical due to operational or privacy constraints.

Second, for \giTMLE $\psi(\bQ_n^\dagger)$ to be asymptotically normal, $g_n$ needs to stabilize such that $\frac{1}{n}\sum_{i=1}^nP_{Q_0,g_i}D(\Qinf,g_i)^2$ converges as $n \to \infty$.
In contrast, our proposed \gbarTMLE $\psi(\bQ_n^*)$ requires only the average design $\gbar_n$ to stabilize. 
Likewise, while the efficiency proof for \gbarTMLE can be extended to \giTMLE, the latter requires stronger positivity and convergence conditions on the adaptive design than those needed for \gbarTMLE.

Finally, suppose both estimators have the same $\Qinf$. By Jensen's inequality we have that
\begin{eqnarray*}
    && \frac{1}{n}\sum_{i=1}^nP_{Q_0,g_i}D(\Qinf,g_i)^2 \\
    &=& P_{Q_{0,W}} \frac{1}{n}\sum_{i=1}^n \frac{E[Y-\bQinf(1,W)|A=1,W]^2}{g_i(1|W)} +  P_{Q_{0,W}}\frac{1}{n}\sum_{i=1}^n \frac{E[Y-\bQinf(0,W)|A=0,W]^2}{g_i(0|W)} \\
    &&+ P_{Q_{0,W}}\left[\bQinf(1,W)-\bQinf(0,W)-\psi(\Qinf)\right]^2\\
    &\geq& P_{Q_{0,W}} \frac{E[Y-\bQinf(1,W)|A=1,W]^2}{\gbar_n(1|W)} +  P_{Q_{0,W}}\frac{E[Y-\bQinf(0,W)|A=0,W]^2}{\gbar_n(0|W)} \\
    &&+ P_{Q_{0,W}}\left[\bQinf(1,W)-\bQinf(0,W)-\psi(\Qinf)\right]^2 \\
    &=& P_{Q_0,\gbar_n}D(\Qinf,\gbar_n)^2.
\end{eqnarray*}

This implies that $\psi(Q_n^*)$ is more efficient than $\psi(Q_n^{\dagger})$ in finite samples. 
\end{remark}
\section{Adaptive Designs Motivated by Efficiency Results}
\label{section:adaptive_design}
Beyond efficient estimation after the completion of a given adaptive experiment, we investigate how to design real-time adaptive experiments to further reduce the variance of the estimator for a specific causal estimand of interest.
Ideally, one can characterize an oracle experimental design that minimizes semiparametric efficiency bound for estimating that estimand using the data collected through that oracle design, which we denote by $\pi^*:\cW \times \cA \to [0,1]$. Using ATE as an example, it is known that the oracle design $\pi^*$ follows the conditional Neyman allocation:
$
\pi^*(a|w) := \frac{\sigma(a,w)}{\sigma(1,w)+\sigma(0,w)}, \label{formula:neyman_allocation}$
where $\sigma^2(a,w):=Var[Y|A=a,W=w]$. 

However, the oracle design is unknown \textit{a priori} in practice. 
This necessitates adaptive experimental designs that sequentially learn the oracle design from the accumulated data and determine treatment randomization probabilities for newly enrolled units according to the evolving estimate of the oracle design, with the aim of improving estimation efficiency.
Motivated by efficiency results of \gbarTMLE, in this section we propose a novel adaptive design that guides the average of adaptive design toward the oracle design, which we refer to ``$\gbar$-driven adaptive design''. 

This design is constructed as follows.
At the beginning of the experiment, we enroll \(n_0\) units and randomly assign treatments with a baseline probability (e.g. 0.5). Here, $n_0$ is a user-specified ``burn-in'' sample size to estimate the oracle design for adapting treatment randomization probabilities. 
After that, suppose we have enrolled $i$ units and define $\gbar_i = \frac{1}{i}\sum_{j=1}^ig_{j}$ as the current average of adaptive design.
Upon the enrollment of each new unit $i+1$, we use the accumulated data $\bO(i)=(O_1,\dots,O_i)$ to construct $\hpi_{i+1}^*$ as an estimate of the oracle design $\pi^*$ for the causal estimand of interest. In the example of ATE, $\hpi_{i+1}^*$ is constructed by first estimating the unknown conditional variance $\sigma^2(a,w)$ by $\hat{\sigma}^2_{1:i}(a,w)$ based on $\bO(i)$, and then plugging that into the Neyman allocation formula:
$
\hpi_{i+1}^*(a|w) = \frac{\hat{\sigma}_{1:i}(a,w)}{\hat{\sigma}_{1:i}(1,w)+\hat{\sigma}_{1:i}(0,w)}.
$
We then integrate this estimate into 
$
\hpibar_{i+1}^* := \frac{1}{i-n_0}\sum_{j = n_0+1}^{i+1} \hpi^*_{j}$ as the target for the next-step average design, which mitigates the variability inherent in adapting solely to the latest estimate and produces a smoother, more stable design trajectory in finite samples.
Finally, by solving $\gbar_{i+1} = \hpibar^*_{i+1}$ for the next average design, the treatment randomization function $g_{i+1}$ is given by
\[
g_{i+1} := (i+1)\hpibar^*_{i+1} - i\, \gbar_i,
\]
clipped to $[0,1]$.
This adaptation procedure is repeated upon enrollment of the new unit.

\begin{remark}[Comparison with standard Neyman-allocation adaptation]
Compared with the adaptive designs with standard Neyman allocation that directly sets $g_{i+1}=\hpi_{i+1}^*$, our $\gbar$-driven adaptive design adopts a ``self-calibrating'' approach:
if the average design in earlier stages presents higher probability to assign treatment (control) relative to the estimated optimal design, the $\gbar$-driven adaptive design compensates in subsequent stages by reducing the randomization probability for treatment (control) accordingly.
This design strategy offers global control over the adaptive design as a whole, steering the overall average design toward the oracle design in a stable manner, rather than optimizing $g_{i+1}$ alone in isolation without considering previous treatment randomization in the adaptive experiment. Together, the $\gbar$-driven design framework and efficient estimation framework of \gbarTMLE work jointly to improve estimation efficiency in answering various causal questions of interest.
\end{remark}

\section{Simulation Study}
\label{section:simulation}
\subsection{Setup}

This section presents simulation studies to evaluate the performance of the proposed estimator and adaptive design to minimize variance in estimating ATE as a representative example.
Specifically, we examine the following aspects: 
\textbf{(A)} Estimation: We evaluate the performance of \gbarTMLE in their efficiency in estimating ATE and compared that with \giTMLE across adaptive designs that either aim to improve estimation efficiency for the causal estimand or target a different objective.
\textbf{(B)} Design: With the goal of minimizing variance in ATE estimation, we evaluate the performance of the proposed $\gbar$-driven design and compare that with other alternative designs.

To answer these questions, we consider five types of experimental designs that span 10 time points, starting with $n_0=1000$ units enrolled at the initially, and 250 units addeded at each subsequent time point. Each design differs in its mechanism to determine treatment randomization probabilities given the accumulated data at each time point. One design is non-adaptive, assigning treatments and controls with equal probability ($g_i \equiv 0.5$ for every $i$) throughout the experiment. 
Another three designs are based on Neyman allocation:
one is a standard adaptive design that directly applies the current estimate of Neyman allocation for treatment randomization for the new enrollee that sets $g_i=\hpi_i^*$; another is our proposed $\gbar$-driven Neyman allocation design; and the third is an ideal oracle design, which applies the oracle Neyman allocation from the outset ($g_i = \pi^*$ for every $i$), as if it were known.
We also implement a ``benefit-driven'' design that estimates conditional average treatment effect (CATE) function $E[Y(1)-Y(0)|W]$ from the accumulated data and tilts the treatment randomization probabilities toward the estimate of optimal personalized treatment $I(E[Y(1)-Y(0)|W]>0)$, thereby increasing the expected outcome of new enrollees compared with that in the fixed non-adaptive design. Each design is repeated across 5000 Monte Carlo runs. More details are deferred to Appendix \ref{app:sim_details}. 

Figure \ref{fig:DGD1-Neyman_Benefit} displays the oracle Neyman allocation and CATE functions across baseline covariates.
The oracle Neyman allocation function varies across covariates, resulting in heterogeneous treatment randomization probabilities, while CATE function is also heterogeneous and follows a different pattern. Consequently, the benefit-driven design is expected to yield higher variance in ATE estimation than the adaptive designs using Neyman allocation.

\begin{figure}[tbh]
    \centering
    \includegraphics[width=0.3\linewidth]{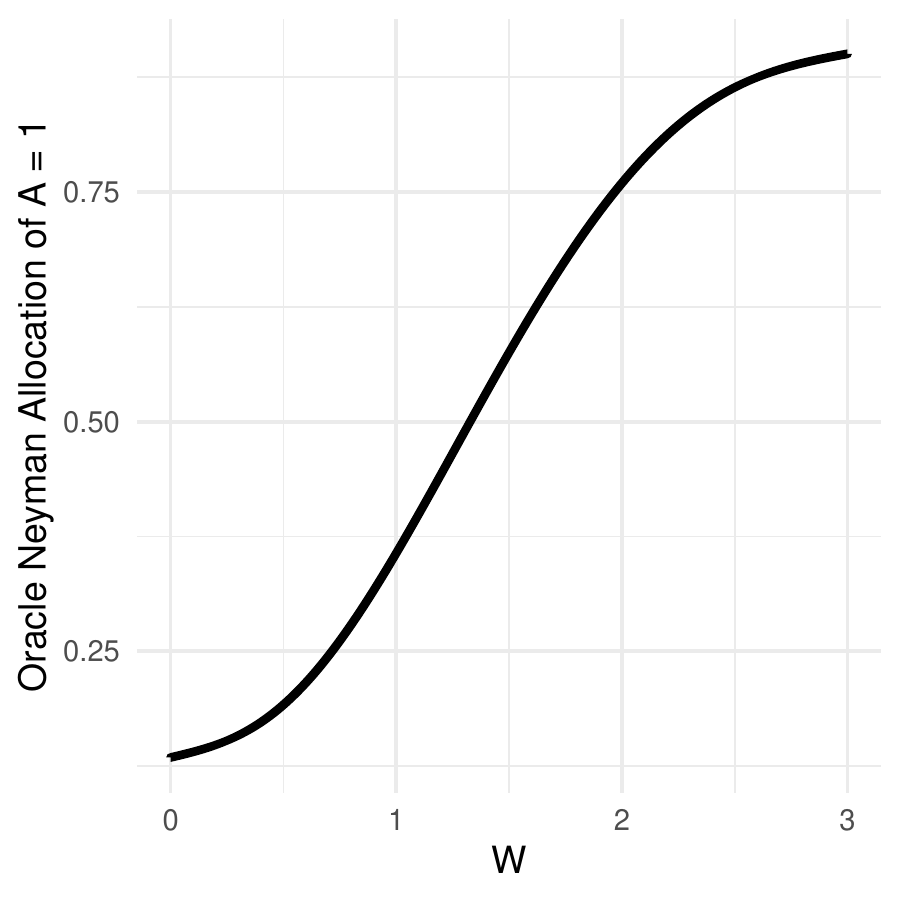}
    \includegraphics[width=0.3\linewidth]{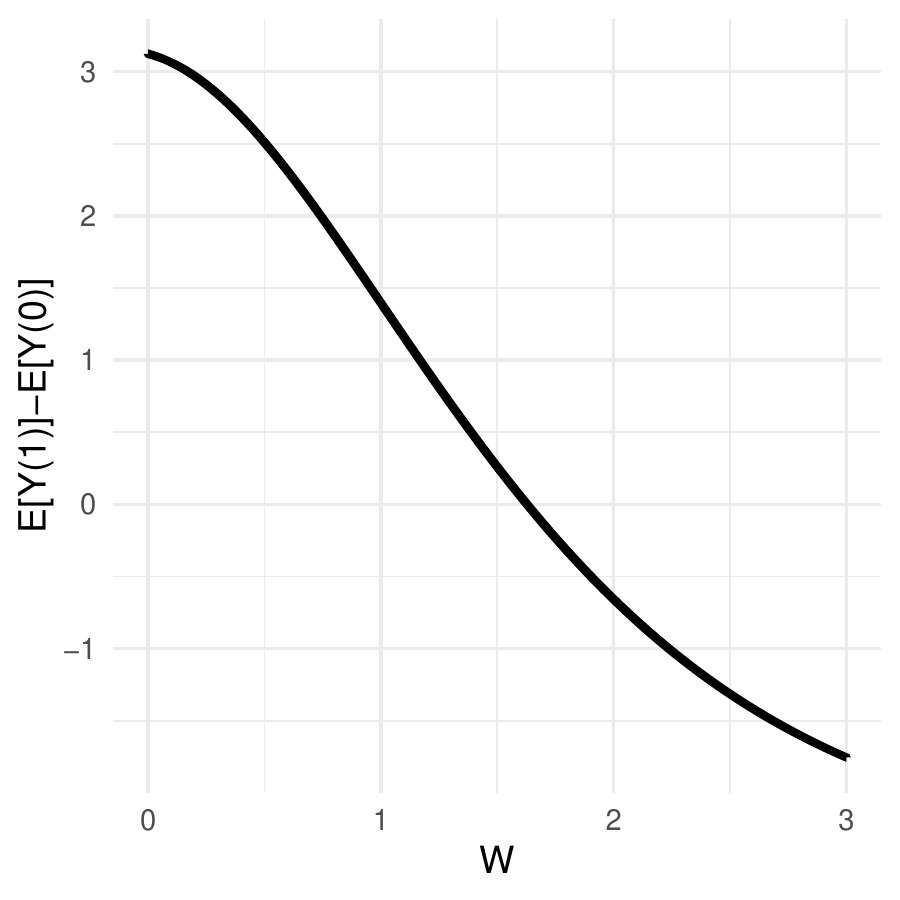}
    \caption{Left: Probability of assigning treatment assignment $a=1$ (y-axis) under the oracle Neyman allocation across baseline covariates (x-axis). Right: Conditional average treatment effect (CATE) (y-axis) across baseline covariates (x-axis).}
\label{fig:DGD1-Neyman_Benefit}
\end{figure}

\subsection{Results}
\textbf{(A) Estimation.}
Table \ref{tab:est_sim_1} presents the bias, variance, and MSE and the coverage rate of the 95\% Wald-type confidence interval of the proposed \gbarTMLE in adaptive experiements with benefit-driven design and standard Neyman-allocation design, respecitvely. It also reports the coverage of the 95\% Wald-type ``oracle confidence interval'' based on the estimator's Monte Carlo standard deviation of the point estimates, with its nominal coverage supporting the estimator's asymptotic normality.
\renewcommand{\arraystretch}{0.6}
\begin{table}[ht]
\centering
\footnotesize
\begin{tabular}{cccccccc}
\toprule
Design & Time &  Bias ($10^{-2}$) & Var. ($10^{-2}$) & MSE ($10^{-2}$) & Cov. (\%) & Oracle Cov. (\%) \\
\midrule
Benefit Driven & 2  &$-0.55$ & $6.22$ & $6.23$ & 93.7 & 94.6 \\
Benefit Driven & 4  &$-1.56$ & $5.37$ & $5.40$ & 94.4 & 94.5 \\
Benefit Driven & 6  &$-2.40$ & $5.00$ & $5.05$ & 95.0 & 94.6 \\
Benefit Driven & 8  &$-3.27$ & $4.82$ & $4.93$ & 94.6 & 94.4 \\
Benefit Driven & 10 &$-2.94$ & $4.91$ & $5.00$ & 94.2 & 94.4 \\ 
\midrule
Standard Neyman & 2  &$1.90$  & $5.49$ & $5.53$ & 92.9 & 94.9 \\
Standard Neyman & 4  &$2.93$  & $3.72$ & $3.81$ & 92.3 & 94.4 \\
Standard Neyman & 6  &$1.16$  & $2.94$ & $2.95$ & 92.0 & 94.7 \\
Standard Neyman & 8  &$1.65$  & $2.30$ & $2.33$ & 92.6 & 94.7 \\
Standard Neyman & 10 &$1.09$  & $2.03$ & $2.04$ & 92.6 & 95.1 \\
\bottomrule
\end{tabular}
\normalsize
\caption{Performance of \gbarTMLE in adaptive experiements with benefit-driven design and standard Neyman-allocation design at time points 2, 4, 6, 8, and 10. The table includes bias, variance, mean squared error (MSE), coverage probability of 95\% Wald-type confidence intervals based on estimated standard errors, along with coverage probability of the 95\% Wald-type ``oracle confidence intervals'' constructed by the Monte Carlo standard deviation of the point estimates.}
\label{tab:est_sim_1}
\end{table}

Figure \ref{fig:DGD1-rel_var_gbar_gi} compares the relative variance of the proposed TMLE estimator \gbarTMLE $\psi(\bQ_n^*)$ with that of \giTMLE $\psi(\bQ_n^{\dagger})$. 
The results show that \gbarTMLE $\psi(\bQ_n^*)$ significantly reduces variance compared with \giTMLE in adaptive experiments, under both benefit-driven design and standard Neyman allocation, demonstrating \gbarTMLE's efficiency across different adaptive experiments with varying design objectives.

\begin{figure}
\centering'\includegraphics[width=0.48\linewidth]{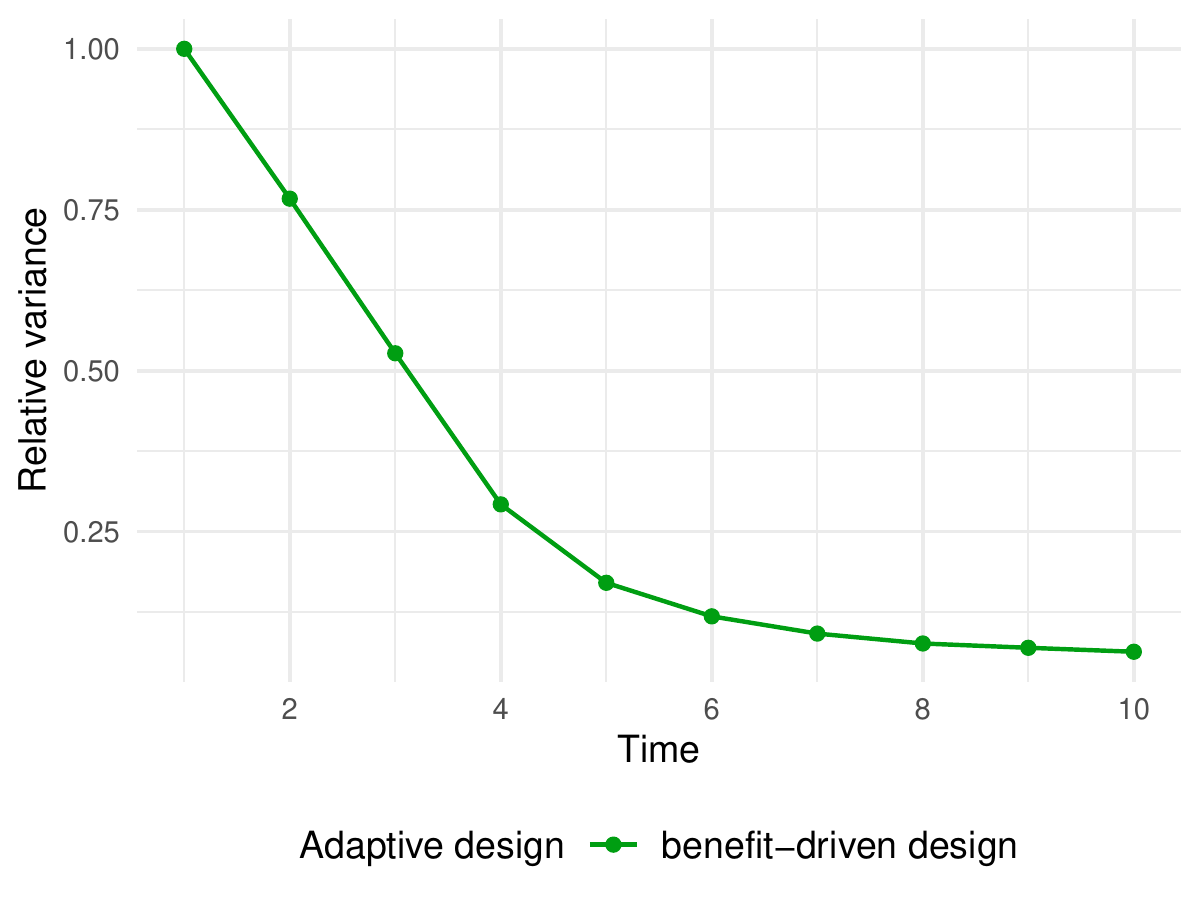}
    \includegraphics[width=0.48\linewidth]{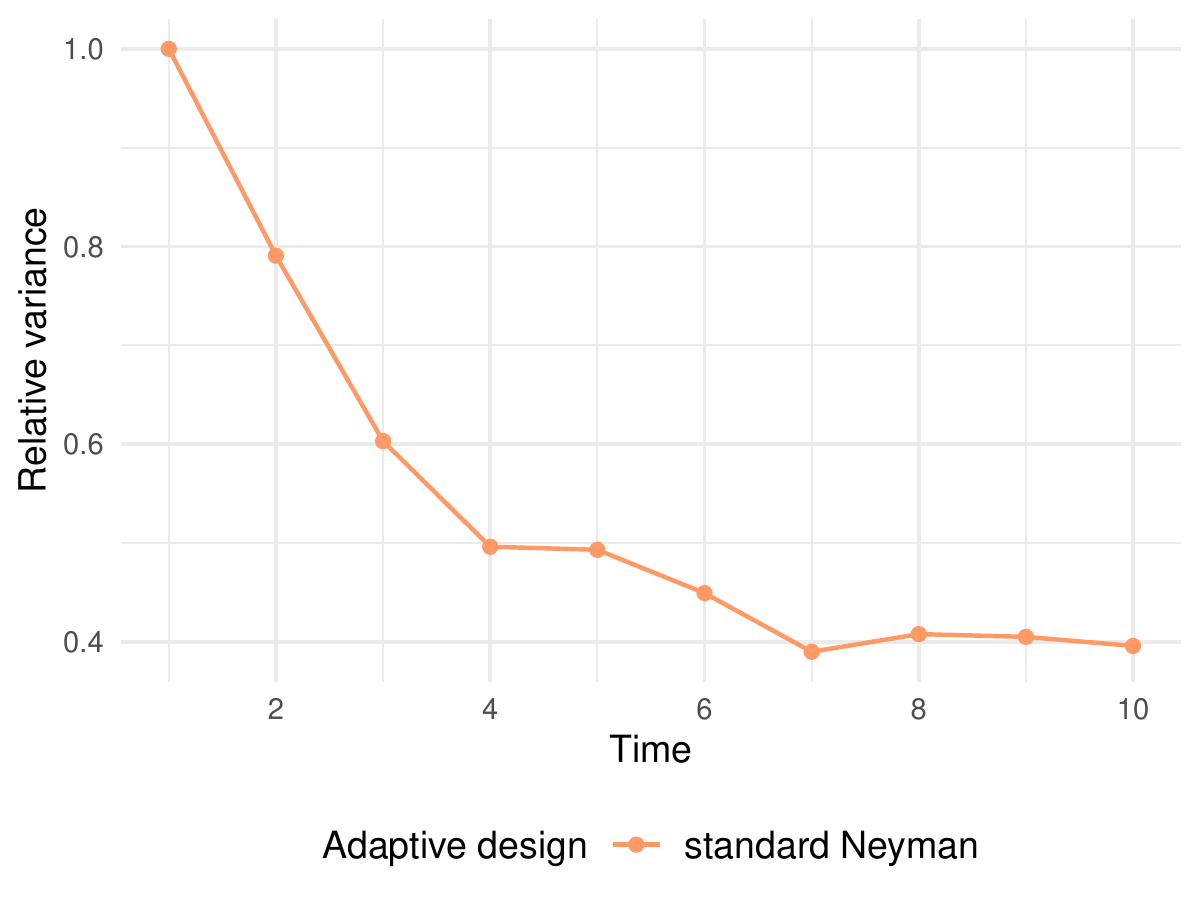}

    \caption{Relative variance of \gbarTMLE $\psi(\bQ_n^*)$ versus \giTMLE $\psi(\bQ_n^{\dagger})$ over the course of adaptive experiments with benefit-driven design (left) and with standard Neyman allocation (right). In each plot, the x-axis is the time point of the adaptive experiment, and the y-axis is the relative variance of the two TMLEs, each constructed by the cumulative data up to that time point.}
    \label{fig:DGD1-rel_var_gbar_gi}
\end{figure}

\textbf{(B) Design.}
Figure \ref{fig:DGD1-design_rel_var} presents the relative variance of \gbarTMLE under various designs compared with that in the non-adaptive design with fixed treatment randomization probability 0.5.
The results show that both adaptive designs with standard Neyman allocation and the proposed $\gbar$-driven Neyman allocation yield lower variance than the non-adaptive design. Moreover, the $\bar{g}$-driven adaptive design consistently results in a smaller variance than the adaptive design with standard Neyman allocation and approaches closer to the oracle design.

\begin{figure}
    \centering
    \includegraphics[width=0.8\linewidth]{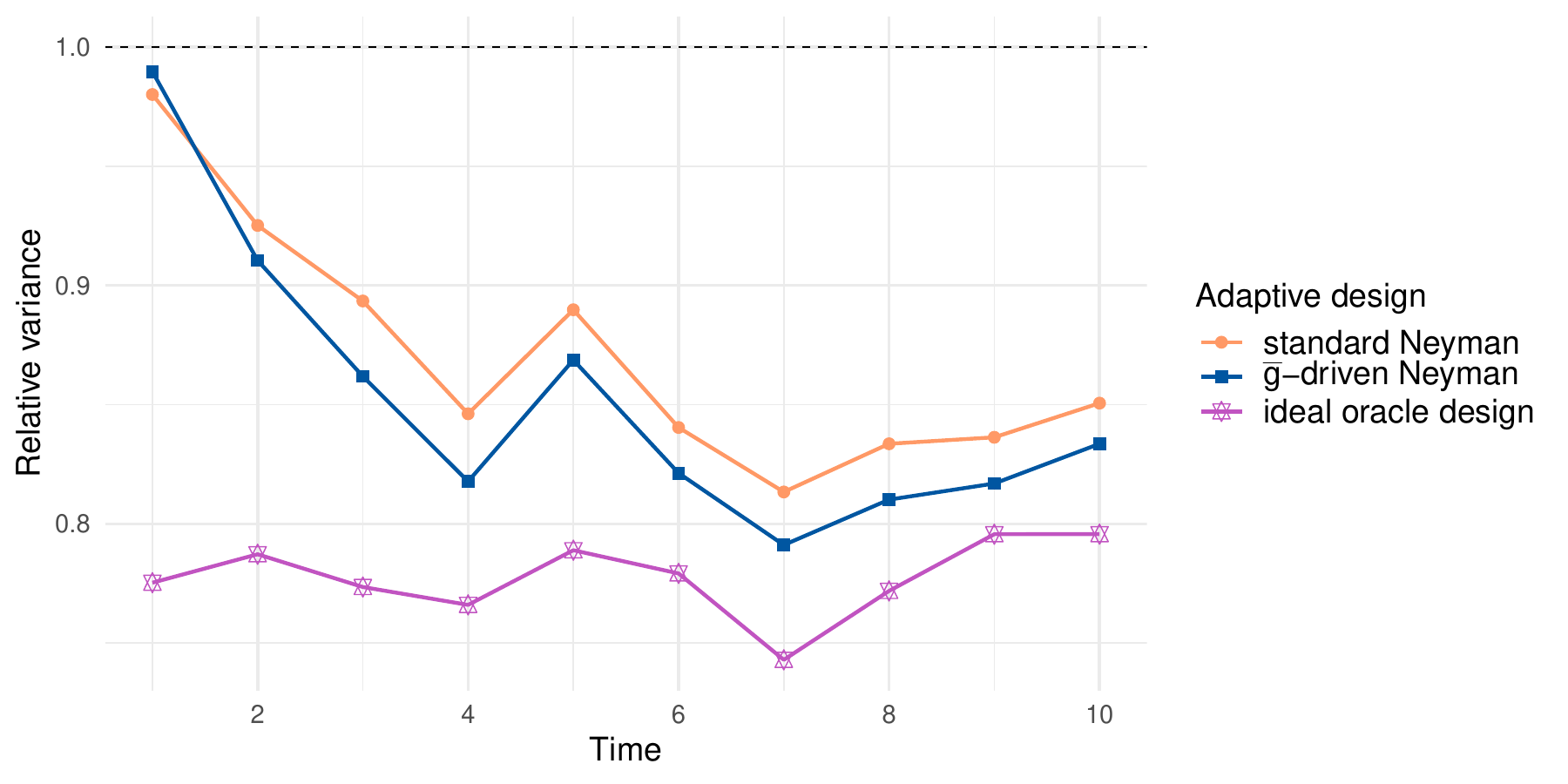}
    \caption{Relative variance of \gbarTMLE in experiments with different designs versus that in the non-adaptive design with a fixed treatment probability of 0.5. Each estimate is constructed using cumulative data up to each time point. The x-axis is the time point of the experiment, and the y-axis is the relative variance.}
    \label{fig:DGD1-design_rel_var}
\end{figure}

Additionally, we evaluate the effects of different designs on semiparametric efficiency for ATE estimation.
Specifically, we construct oracle versions of Neyman-allocation-based adaptive designs, in which the treatment randomization probabilities are based on the oracle Neyman allocation functions using the true conditional variances rather than the estimated ones for units $i \ge n_0$.
Figure \ref{fig:DGD1-squaredEIC_Q0} plots the trajectory of the relative value of $D(Q_0,\gbar_n)^2$ over time for these designs (including the non-adaptive one) against the ideal oracle design that applies the oracle Neyman allocation from the outset. This illustrates how quickly the variance of \gbarTMLE under an adaptive design approaches the semiparametric efficiency bound minimized by the oracle Neyman allocation.
The results show that this quantity under the $\bar{g}$-driven adaptive design converges to the oracle one more rapidly than the standard Neyman allocation, highlighting the advantage of guiding the global average design toward the optimal target.

\begin{figure}
    \centering
    \includegraphics[width=0.8\linewidth]{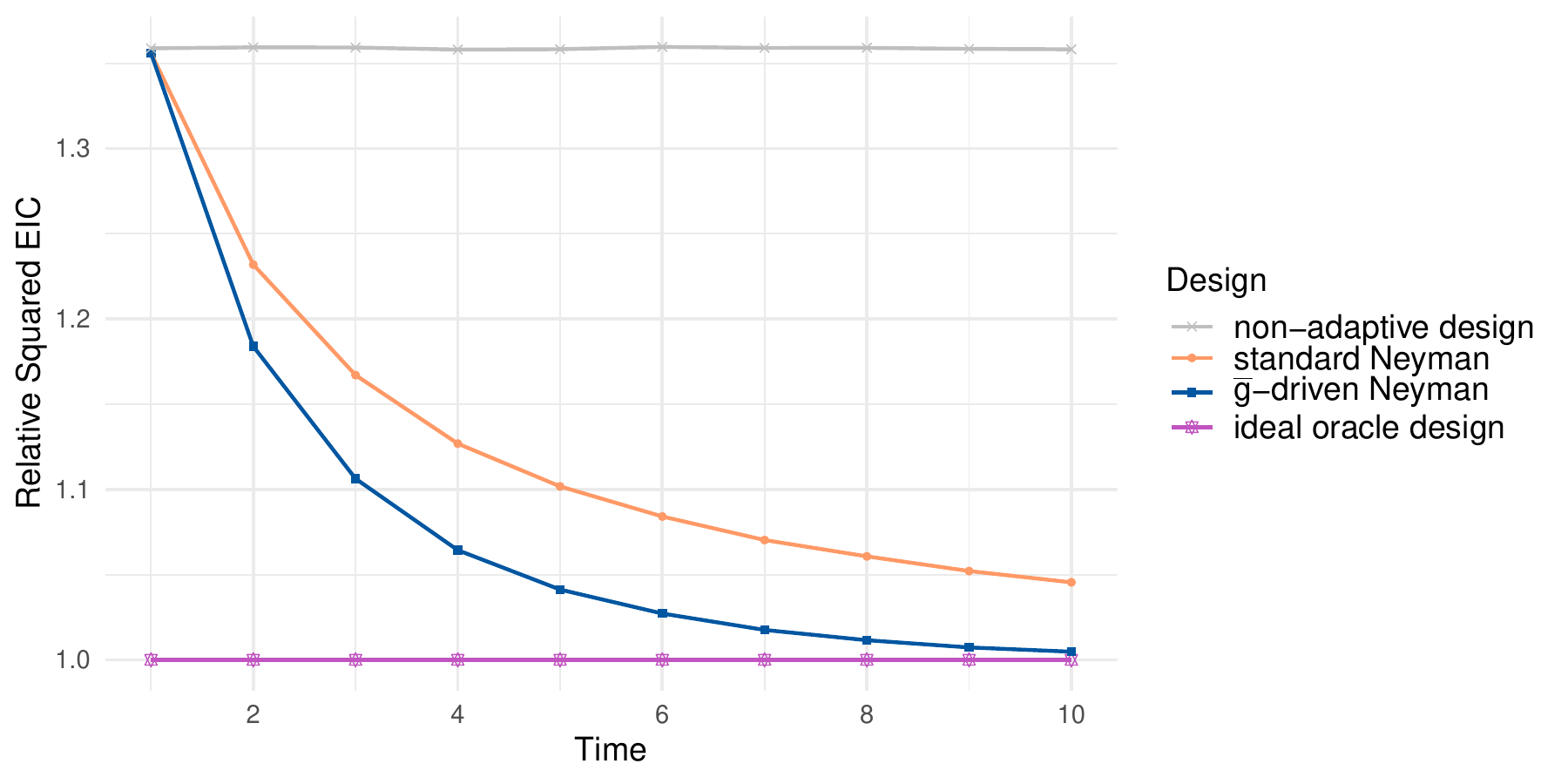}
    \caption{Relative value of $D(Q_0,\bar{g}_n)^2$ under various designs versus that under the ideal design with oracle Neyman allocation (y-axis) across different time points (x-axis). The lines show the trajectories of this quantity over time under non-adaptive design, the oracle version of the adaptive design using standard Neyman allocation and the oracle version of $\bar{g}$-driven Neyman allocation.}
    \label{fig:DGD1-squaredEIC_Q0}
\end{figure}

\section{Application and Generalization}
\label{section:generalization}
\subsection{Application}
Our proposed framework is directly applicable for cases with independent data structure. For example, consider estimating causal estimands from data collected across multiple sites where the site index does not affect the primary outcome but introduces heterogeneity in treatment distribution across sites.
Specifically, define the full data of interest as $X = (W,Y(a)) \sim P^F \in \cM^F$, where $W$ denotes the baseline covariates; $Y(a)$ is the counterfactual outcome with treatment level $a \in \cA$. 
Let the average treatment effect $\Psi^F(P^F)=E_{P^F}[Y(1)-Y(0)]$ be the causal estimand of interest.
Suppose one observe $n$ independent observations $O_1,\cdots,O_n$ with $O_i = (S_i,W_i,A_i,Y_i)$, where $W_i \sim Q_W$ and $Y_i \sim Q_Y(\cdot|A_i,W_i)$ for $Q_W \in \cQ_W$ and $Q_Y \in \cQ_Y$. For every unit $i$, $S_i \in \{1, \dots, J\}$ is the site index that determines the treatment distribution of $A_i \sim g_i(\cdot|W_i)$ with $g_i = g_{s_i}$. Each site $j \in [J]$ has its own treatment assignment mechanism $g_j:\cW \times \cA \to [0,1]$.
Note that each site does not impact the outcome given treatment and baseline covariates, and
$\Psi^F(P^F)$ can be identified by the target estimand $\psi(Q)=E[\bQ(1,W)-\bQ(0,W)]$ under ignorability and positivity assumptions, where $\bQ(A,W) := E[Y|A,W]$.

Define $\gbar_n := \frac{1}{n}\sum_{i=1}^n  g_i$. With the convergence of $\gbar_n$ to a fixed $\gbarinf$ (e.g. $\gbarinf = \sum_{j \in J} p_s(j) g_j$ for a probability mass function $p_s$ of the site index), one can construct an \gbarTMLE by solving the score equation $\frac{1}{n}\sum_{i=1}^n \frac{2A_i-1}{\gbar_n(A_i|W_i)}[Y_i-\bQ_n^*(A_i,W_i)]=0$. This yields an asymptotically normal and efficient estimator that inherits the advantages of \gbarTMLE established for adaptive design settings. This approach also applies to other settings where the site index $S$ is replaced by other variables that do not directly affect the outcome given treatment and other baseline covariates, but only influence the treatment assignment probabilities.

\subsection{The General Framework}
We present the general estimation and design framework for various causal estimands in broader settings, including those with longitudinal data structure.

\textbf{Data, Model and Target Estimand.}
Consider a longitudinal data structure $O = \left(L(0),A(0),L(1),A(1),\cdots,L(K),A(K),L(K+1) \equiv Y\right)$. Let 
\[
\cM(g)=\{ p_{q,g} = (q,\,g) : q \in \cQ \}
\]
be a model with a known \(g\), where 
$
q(o) = \prod_{t=0}^{K+1} q_t(o)$ and $g(o) = \prod_{t=0}^k g_t(o)
$, which corresponds to an i.i.d. likelihood with each $L(t) \sim q_t(\cdot|Pa(L(t)))$ and $A(t) \sim g_t(\cdot|Pa(A(t)))$, where $Pa(L(t))$ and $Pa(A(t))$ denote the set of nodes that appear before $L(t)$ and $A(t)$ in the longitudinal structure of $O$, respectively.
Let $g^* \in \cG^*$ denote a fixed function that defines the causal estimand of interest as $\Psi^F(Q)=E Y_{g^*} = \int_o y\,q(o)\,g^*(o)\,d\mu(o)$. This parameter only depends on the $Q$-component of the likelihood.

Suppose we have \(n\) observations \(\bO(n) = (O_i : i=1,\ldots,n)\) with likelihood
\[
p^n_{q,\bg_{1:n}}(\bo(n)) = \prod_{i=1}^n q(o_i)\,g_i(o_i) = \prod_{i=1}^n \prod_{t=0}^{K+1}  q_t(o_i)  \prod_{t=0}^{K}g_{i,t}(o_i).
\]
This defines the model
$
M^n = \{ p^n_{q,\bg_{1:n}} : q\},
$
where $g_i$ may be specified as known, externally controlled, or partly unknown but estimable from the data.

\textbf{Semiparametrically Efficient Estimation.}  
Define $\gbar_n = \frac{1}{n} \sum_{i=1}^n g_i$ as the average of the whole product \(g_i=\prod_t g_{i,t}\).
Let $D(Q, g)(O_i)$ denote the efficient influence curve in the i.i.d. model $\cM(g)$.
To estimate $E Y_{g^*}$ from data generated by the model $\mathcal{M}^n$, it is useful to consider a reference i.i.d.\ model $\mathcal{M}(\gbar_n)$ with the average design $\gbar_n$. In that i.i.d.\ model, the efficient influence curve for the parameter $\Psi(Q)$ is given by $D(Q, \gbar_n)(O)$. Importantly, any TMLE constructed by solving $n^{-1} \sum_{i=1}^n D(Q, \gbar_n)(O_i)=0$ under the i.i.d.\ model with $\gbar_n$ is also an efficient estimator for the original model $\mathcal{M}^n$ for the observed data, with its asymptotic normality supported by Martingale theories and its semiparametric efficiency provided by the convolution theorem, as shown in Sections \ref{section:asy_normality} and \ref{section:efficiency}. 
This unified estimation framework applies not only to adaptive design settings discussed above, but also to other contexts with longitudinal data structure with dynamic interventions \citep{robins1986new,van2012targeted} and time-series structures \citep{van2018robust,malenica2021adaptive}. For the design problem, the construction of the self-calibrating $\gbar$-driven adaptive design implied by this efficient estimation approach directly mirrors that in Section~\ref{section:adaptive_design}.

\section{Conclusions}
\label{section:conclusion}
In this work, we propose a novel \gbarTMLE framework for estimating a wide class of causal estimands from sequential adaptive experiments.
We establish both asymptotic normality and semiparametric efficiency of \gbarTMLE, and demonstrate superior variance-reduction performance in finite samples. Notably, our framework substantially relaxes the requirements for adaptive experiments that enable valid statistical inference, as it only requires positivity and stabilization for the average design.
Motivated by these results, we further propose a novel $\gbar$-driven adaptive design that dynamically steers the average design toward the estimated optimal design in a stable self-calibrating manner. This adaptive design achieves faster convergence to the oracle design and yields improved estimation efficiency compared with standard approach.
Finally, we provide a general estimation and design framework applicable to broader settings with longitudinal structures.
Future research could explore extensions to other related settings and broader areas of applications.

\newpage
\bibliography{main}

\appendix

\section*{Appendix}

\section{Proofs for Sections \ref{section:asy_normality} and \ref{section:efficiency}}
\label{app:proofs}

\begin{proof}[Proof of Theorem \ref{theorem:tmle_decomposition}]
By applying von Mises expansions with respect to $P_{Q,\gbar_n}$ and $P_{Q_0,g_i}$ across $g_i \in \bg_{1:n}$, we have that
\begin{eqnarray*}
    &&\psi(Q) - \psi(Q_0) \\
    &=&  \frac{1}{n}\sum_{i=1}^n \int D(Q,\gbar_n)(o) d\left(P_{Q, \gbar_n} - P_{Q_0,g_i}\right)(o) \\
    &&+  \frac{1}{n}\sum_{i=1}^n \int\sum_{a=0}^{1}(2a-1)\left(\frac{g_i(a|w)-\gbar_n(a|w)}{\gbar_n(a|w)}\right) \left(\bQ(w)-\bQ_0(w)\right) dQ_W(w) \\
    &=& -\frac{1}{n}\sum_{i=1}^n P_{Q_0, g_i} D(Q,\gbar_n) \\
    &=& -P_{Q_0, \gbar_n} D(Q,\gbar_n).
\end{eqnarray*}
Since \gbarTMLE solves $\frac{1}{n}\sum_{i=1}^n D(Q_n^*,\gbar_n)(O_i)=0$, we have 
\[
\psi(Q_n^*) - \psi(Q_0) = \frac{1}{n}\sum_{i=1}^n [D(Q_n^*,\gbar_n)(O_i) - P_{Q_0, \gbar_n} D(Q_n^*,\gbar_n)].
\] 
By adding and subtracting $D(\Qinf,\gbar_n)(O_i)$, $D(\Qinf,\gbarinf)(O_i)$,
and their conditional expectations given $\bO(i-1)$, we obtain the desired decomposition equation.
\end{proof}

\begin{proof}[Proof of Theorem \ref{theorem:our_LAN}]
The log likelihood ratio of the perturbed distribution and the original distribution follows:
\begin{eqnarray*}
    &&\log\frac{dP^{n}_{Q_{\epsilon_n}^h}(O_1,\cdots,O_n)}{dP^n_{Q}(O_1,\cdots,O_n)} \nonumber \\
    &=& \frac{1}{\sqrt{n}}\sum_{i=1}^n h(O_i) - \frac{1}{2n} \sum_{i=1}^n h^2(O_i) + o_{P^n_{Q}}(1) \\
    &=& \frac{1}{\sqrt{n}}\sum_{i=1}^n h(O_i) - \frac{1}{2n} \sum_{i=1}^n E_{P^n_{Q}}\left[h^2(O_i) \given \bO(i-1)\right] + o_{P^n_{Q}}(1) 
    \label{eq:LAN-average_of_cond_second_moment-h_new}\\
    &=& \frac{1}{\sqrt{n}}\sum_{i=1}^n h(O_i) - \frac{1}{2}\norm{h}_{\cH(Q,\gbar_n)}^2 + o_{P^n_{Q}}(1)     \label{eq:our_LAN_raw_notaion-h_new} \\
    &=& \Delta_{n,h} - \frac{1}{2}\norm{h}_{\cH(Q,\gbarinf)}^2,
\end{eqnarray*}
where we denote $\Delta_{n,h} := \frac{1}{\sqrt{n}}\sum_{i=1}^n h(O_i) + o_{P^n_{Q}}(1)$, with the first term a zero-mean martingale. 
The second equation is based on the convergence of conditional variance of the zero-mean martingale according to Theorem 2.23 of \cite{hall2014martingale}, as $E[h^2(O_i)|\bO(i-1)]$ is bounded.
The third equation is based on the dependence of $g_i$ on the past $\bO(i-1)$ in our adaptive design setup and the definition of $\norm{h}_{\cH(Q,\gbar_n)}^2$.
The last equation relies on the convergence of the adaptive design described in Assumption \ref{assumption:Hilbert_space_approximation} and results in $\Delta_{n,h} \dto \Delta_{h} \sim N(0, \norm{h}_{\cH(Q,\gbarinf)}^2)$ by Martingale central limit theorem.
\end{proof}

\begin{proof}[Proof of Theorem \ref{theorem:efficiency}]
By Assumptions \ref{assumption:strong_positivity_of_average_design},
\ref{assumption:stabilized_variance_gbarinf},  \ref{assumption:reasonable_covering_integral}, \ref{assumption:sigma_N_convergence},
\ref{assumption:reasonable_covering_integral_gbar}, and \ref{assumption:sigma_N_convergence_gbar}, one can apply Theorem \ref{theorem:asympotic_normality_proposed_tmle} to establish asymptotically normality of $\psi(Q_n^*)$ under $P^n_{Q_0}$: 
    \[
    \sqrt{n}[\psi(Q_{n}^*)-\psi(Q_0)] \dto N\left(0,P_{Q_0,\gbarinf}D(Q_0,\gbarinf)^2\right).
    \] 
    With Assumption \ref{assumption:Hilbert_space_approximation}, the log likelihood ratio for the perturbed distribution of $P^n_{Q_0}$ is obtained by $log\frac{dP^{n}_{Q_{\epsilon_n}^h}(O_1,\cdots,O_n)}{dP^n_{Q_0}(O_1,\cdots,O_n)} = \Delta_{n,h}-\frac{1}{2}\norm{h}_{\cH(Q_0,\gbarinf)}^2$, where $h \in \cH(Q_0,\gbarinf)$ and $\Delta_{n,h} \dto \Delta_{h} \sim N(0, \norm{h}_{\cH(Q_0,\gbarinf)}^2)$ under $P^n_{Q_0}$ via Theorem \ref{theorem:our_LAN}.    
    Then, by Le Cam's third lemma (Theorem 3.10.7 in \cite{van1996weak}), we have that
    \[
    \sqrt{n}[\psi(Q_n^*)-\psi(Q_0)] \dto N(<D(Q_0,\gbarinf),h>, P_{Q_0,\gbarinf}D(Q_0,\gbarinf)^2)
    \]
    under $P^{n}_{Q_{\epsilon_n}^h}$.
    Combining with the regularity of the parameter 
    \[
    \sqrt{n} \left(\psi(Q_{\epsilon_n}^{h})-\psi(Q_0)\right) \to \dot{\psi}(Q_0)(h),
    \]
    where $\dot{\psi}(Q_0)(h)=<D(Q_0,\gbarinf),h>$,
    we have that 
    \[
    \sqrt{n} \left(\psi(Q_n^*)-\psi(Q_{\epsilon_n}^{h})\right) \dto N(0,P_{Q_0,\gbarinf}D(Q_0,\gbarinf)^2)
    \]
    under $P^{n}_{Q_{\epsilon_n}^h}$. Hence $\psi(Q_n^*)$ is regular.
    
    By the convolution theorem (Theorem \ref{theorem:LAN_conv}), the asymptotic variance of $\sqrt{n} (T_n-\psi(Q_0))$ for regular estimators $T_n$ is no less than $\norm{\dot{\psi}^*(Q_0)}_{\cH(Q_0,\gbarinf)}^2$, which equals $P_{Q_0,\gbarinf}D(Q_0, \gbarinf)^2$ and is achieved by the asymptotic variance of $\psi(Q_n^*)$. Therefore, $\psi(Q_n^*)$ is efficient.
\end{proof}

\section{Simulation Details}
\label{app:sim_details}
The data generating distribution is as follows:
$
\bQ_0(1,W) = 25 + 10 /\left[1 + exp(-(2W + 1))\right]$;
$
\bQ_0(0,W) = 20 + 17.5/\left[1 + exp(-(W + 0.1))\right]$;
$
Var_{Q_0}[Y|A=1,W]=1 + 3W^3$;
$
Var_{Q_0}[Y|A=0,W]=1 + 1.5\times(3 - W)^3.
$
Figure \ref{fig:DGD_1} shows the true conditional expectation and variance of the outcome given different treatment levels and baseline covariates.
\begin{figure}[tbh]
  \centering
  \includegraphics[width=0.95\textwidth, height=0.6\textwidth]{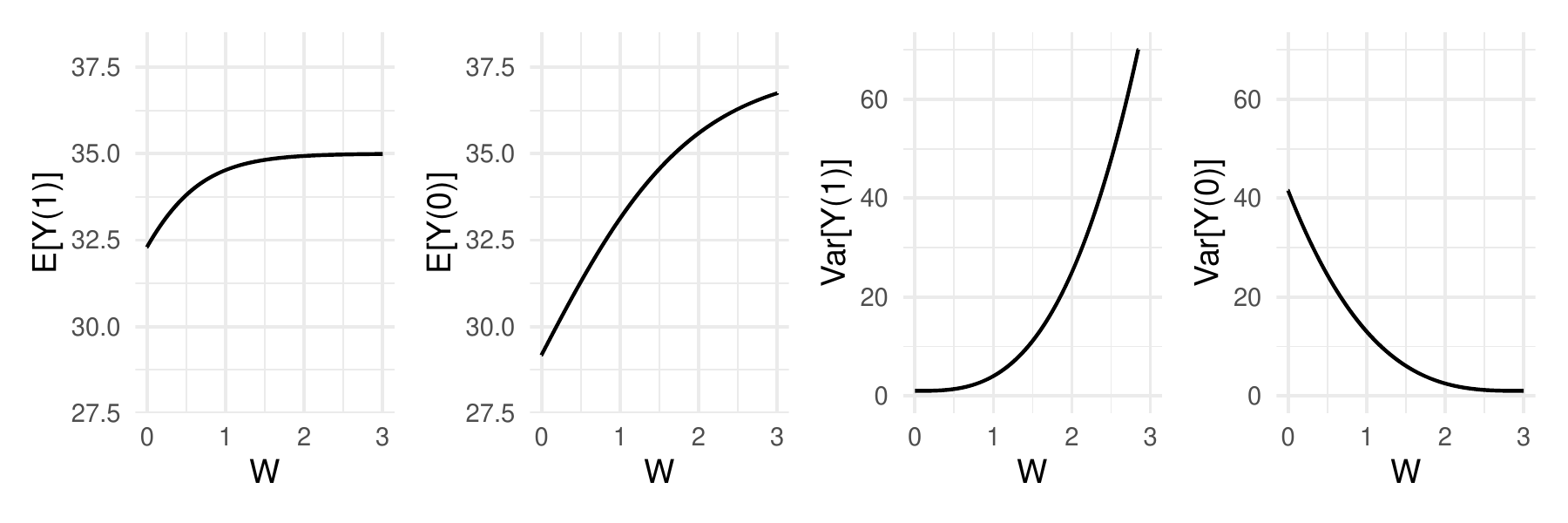}
    \caption{Summary plots of the data generating distribution. The first two plots show the conditional mean of the outcome (y-axis) for treatment ($a=1$) and control ($a=0$) against covariates $W$ (x-axis). The last two plots present the corresponding conditional variances.}
    \label{fig:DGD_1}
\end{figure}
For benefit-driven design, we guide treatment randomization probabilities toward optimal personalized treatment using baseline covariates through estimating the conditional average treatment effect (CATE). Define the true CATE function as $B_{0}: \cW \to \R$, with $B_0(W)=\bQ_0(1,W)-\bQ_0(0,W)$. For a newly enrolled unit $i$ with baseline covariates $w_i$, we use $B_{n}(w_i)$ to denote the estimate of $B_{0}(w_i)$ based on the current data $\bO(i-1)$.
We estimate CATE functions by first transforming the outcomes into doubly robust pseudo-outcomes \citep{rubin2005general,van_der_laan_targeted_2006,van2015targeted,kennedy2023towards}, then regressing them on baseline covariates using Super Learner \citep{van2007super}.
This estimated CATE guides the treatment randomization probabilities toward the estimated optimal personalized treatment $I(B_{n}(w_i)>0)$ through a function $\Gamma_{\nu,b}$ introduced by \cite{chambaz2017targeted}:
\begin{eqnarray*}
\Gamma_{\nu,b}(x)=\nu I(x \leq -b) + (1-\nu) I(x \geq b) + \left(-\frac{1/2 - \nu}{2b^3}x^3 + \frac{1/2 - \nu}{2b/3}x + \frac{1}{2}\right) I\left(-b \leq x \leq b\right).
\end{eqnarray*} 
Here, $B_{n}(w_i)$ is used as input to $\Gamma_{\nu,b}$ to determine the probability of assigning treatment $a=1$; $\nu \in [0,0.5]$ sets a lower bound on the treatment randomization probability, ensuring that each arm is chosen with probability at least $\nu$; $b>0$ controls the range within which these probabilities respond smoothly to the estimated CATE $B_{n}(w_i)$.
When $B_{n}(w_i)>b$, indicating a substantially higher expected outcome under treatment $a=1$ compared with $a=0$, the treatment allocation probability of $a=1$ reaches the maximum $1-\nu$; when $B_{n}(w_i)<-b$, the treatment allocation probability of $a=1$ drops to its minimum value $\nu$. When $B_{n}(w_i) \in [-b,b]$, that probability varies smoothly with the value of CATE. As the estimated CATE approaches to 0.5, the probability gets closer to 0, balancing exploration and exploitation.
In simulation, we set $b=1$ and $\nu_i = exp(-t_i)$ for every unit $i$, where $
t_i= \left\lceil \frac{i - n_0}{250} \right\rceil$ with $n_0 = 1000$, thereby increasing the chance of assigning substantially better treatments as the experiment progresses.

\end{document}